\documentclass[a4paper]{amsproc}
\usepackage{amssymb}
\usepackage{amscd}
\usepackage[dvips]{graphicx}
\usepackage[all,cmtip]{xy}

 \newtheorem{thm}{Theorem}[section]
 
 \newtheorem{lem}{Lemma}[section]

 \newtheorem{exm}{Example}[section]
 \newtheorem{rem}{Remark}[section]

\numberwithin{equation}{section}

\def\ji {\char'032}

\def\m  {\char'176}

\font\rrm=wncyr8%
\font\rit=wncyi8

\newcommand{\A}{A^{-1}}
\newcommand{\R}{\mathbb{R}}

\DeclareMathOperator{\Ad}{\mathrm{Ad}}
\DeclareMathOperator{\Span}{\mathrm{span\,}}
\DeclareMathOperator{\diag}{\mathrm{diag}}

\title{Heisenberg model in pseudo--Euclidean spaces II}

\author[Bo\v zidar Jovanovi\'c \and Vladimir Jovanovi\'c]{}

\email{bozaj@mi.sanu.ac.rs}
\email{vlajov@blic.net}

\begin{document}

\maketitle

\centerline{\scshape Bo\v zidar Jovanovi\'c}
\medskip
{\footnotesize
\centerline{Mathematical Institute SANU}
\centerline{Serbian Academy of Sciences and Arts}
\centerline{Kneza Mihaila 36, 11000 Belgrade, Serbia}
}

\medskip

\centerline{\scshape Vladimir Jovanovi\'c}
\medskip
{\footnotesize
 \centerline{Faculty of Sciences}
   \centerline{University of Banja Luka}
   \centerline{Mladena Stojanovi\'ca 2, 51000 Banja Luka, Bosnia
and Herzegovina}
}

\bigskip

\begin{abstract}
In the review we describe a
relation between the Heisenberg spin chain
model on pseudospheres and light--like cones in pseudo--Eucli\-dean spaces
and virtual billiards.
A geometrical interpretation of the integrals associated to a family of confocal quadrics is given, analogous to
Moser's geometrical interpretation of the integrals of the Neumann system on the sphere.
\end{abstract}

\tableofcontents

\section{Introduction}

In the paper we round up
our study of geometry of discrete (contact)
integrable systems with constraints starting with the Heisenberg system in pseudo--Euclidean spaces $\mathbb E^{k,l}$ (see \cite{Jov2}) and continued with billiard system within ellipsoid \cite{JJ}, i.e, virtual billiard system within quadrics in $\mathbb E^{k,l}$ \cite{JJ3}.\footnote{A draft of Section 2 of the current paper is given as the Section 5 in the first arXive version of \cite{JJ3} [arXiv:1510.04037v1].}

It is well known that the Heisenberg system on a sphere can be seen as a square root of the ellipsoidal billiard  \cite{MV, Ves3}, as well that it can be seen as a B\"acklund transformation of the Neumann system \cite{Su}. For the latter, Moser gave a nice geometrical interpretation of integrability (e.g., see \cite{Mum}). We feel that it would be interesting to formulate analogous pseudo-Euclidean statements. In this sense we compiled a review paper, with some additional analysis concerning mostly the light--like case.  We note that integrable discretizations are  usually considered for complexified objects. Here we work within real domains. For example, the Moser--Veselov skew hodograph mapping naturally follows from the requirement that a quadratic generating function defines a symplectic mapping for real objects (see Lemma \ref{karakterizacija}). As an example we obtain the symplectic billiard mapping for the ellipsoid recently introduced in \cite{AT} (see Example \ref{ST0}).

We consider the Heisenberg model on a pseudosphere
(light--like cone)
\[
S_c^{n-1}=\{q\in\mathbb E^{k,l}\,\vert\,\langle q,q\rangle=c\}, \qquad c=\pm 1, 0
\]
in a pseudo--Euclidean space $(\mathbb E^{k,l}, \langle\cdot,\cdot\rangle)$ of signature $(k,l)$, $k+l=n$ (see \cite{Jov2}).
It
is defined as a discrete Lagrangian system given by the action
functional
\[
\mathrm S[\mathbf q]=\sum \mathbf{L}(q_k,q_{k+1}), \qquad \mathbf
L(q_k,q_{k+1})=\langle q_{k},J q_{k+1}\rangle,
\]
where $\mathbf q=(q_k), \, k\in\mathbb Z$ is a sequence of points
on $S^{n-1}_c$ and $J=\diag(J_1,\dots,J_{n})$, $\det J\ne 0$.
In the Euclidean case the functional defines the energy of a spin
chain of the Heisenberg model, see Veselov
\cite{Ves3}.

The equations of the stationary configuration
have the form
\begin{equation}\label{neumann}
\frac{\partial \mathbf L(q_k,q_{k+1})}{\partial q_k}+
\frac{\partial \mathbf L(q_{k-1},q_{k})}{\partial q_k}=
EJ q_{k+1}+EJ q_{k-1}=\lambda_k E q_k, \qquad k\in\mathbb Z,
\end{equation}
where\footnote{We hope that it will be clear from the context when $k$ denotes the discrete time, and when the
signature of the metric.}
\[
E=\diag(\tau_1,\dots,\tau_n), \quad \tau_i=1, \quad i=1,\dots,k, \quad \tau_i=-1,\quad i=k+1,\dots,n.
\]

The multipliers
\begin{equation}\label{lambda}
\lambda_k=2\langle J^{-1} q_k,q_{k-1}\rangle/\langle J^{-2}
q_k,q_k\rangle
\end{equation}
are determined by the
constraints $\langle q_k,q_k\rangle=c$, and they are
defined outside the singular set $\langle J^{-2} q_k,q_k
\rangle=0$.

The equations \eqref{neumann}, \eqref{lambda} determine the
symplectic mapping
\[
\Phi\colon P_c\to P_c, \qquad \Phi(q_{k-1},q_k)=(q_k,q_{k+1})
\]
with respect to the 2-form $\Omega=\sum_i \tau_i J_i dQ_i\wedge dq_i$, where
\begin{align*}
P_c(q,Q)\colon  \quad & \langle q,q\rangle=c, \qquad\quad
\langle Q,Q\rangle=c,\qquad\,\, c=\pm 1, 0,\\
&\langle q,J^{-1}Q\rangle\ne 0, \quad \langle q,J^{-2}q\rangle\ne 0, \quad \langle Q,J^{-2}Q\rangle\ne 0
\end{align*}
(see \cite{Jov2}).\footnote{Actually, the function $\langle q_k,J^{-1}q_{k+1}\rangle$ is the first integral \cite{Jov2}, and so the
condition $\langle q_k,J^{-1}q_{k+1}\rangle\ne 0 $ is invariant of the dynamics, while
$\langle J^{-2} q_k,q_k
\rangle\ne 0$ is not. If $\langle J^{-2} q_{k+1},q_{k+1}
\rangle=0$, by definition the flow stops. In this sense, in the codomain of $\Phi$ we should take the manifold defined without the assumption $\langle Q,J^{-2}Q\rangle\ne 0$.}
It is a completely integrable discrete Hamiltonian system.
For $J^2_j\ne J^2_i$, the integrals can be written in the form
\begin{equation}\label{int-neumann}
f_i(q_{k-1},q_{k})=c\cdot \tau_i (q_{k-1})_i^2+\sum_{j\neq
i}\frac{\tau_i\tau_j((Jq_k)_j(q_{k-1})_i-(q_{k-1})_j(Jq_k)_i)^2}{
J^2_i- J^2_j},
\end{equation}
$i=1,\dots,n$, with the relation $\sum_i f_i \equiv c^2$ among them.
Furthermore,
on the light--like cone, the mapping $\Phi$ leads to an integrable contact system
as well (see \cite{Jov2}).

\subsection*{Outline and results of the paper.}
In the Euclidean case, there is a remarkable relation between the
ellipsoidal billiard and the Heisenberg spin chain model established by the
use of so the called skew hodograph mapping (see Moser and Veselov
\cite{MV}). Recently, in \cite{Jov3}, a simple observation concerning generating functions for systems with constraints (see Theorem \ref{prva})
is used for another interpretation of the skew-hodograph mapping.
Following \cite{Jov3}, we establish analogous relation between virtual billiards and the pseudo--Euclidean Heisenberg model, which also includes the symmetries of the system  (Theorem \ref{druga}, Section 2).
As a by-product, we obtain the symplectic billiard within ellipsoid given in \cite{AT} (Example \ref{ST0}), as well as a "big" $n\times n$--matrix representations of the
virtual billiard flow
(Theorem \ref{treca}, Section 2).

Further, in Sections 3 and 4, as a straightforward generalization of the Euclidean case (see \cite{Su}),  we consider a discrete Legendre transformation of the Heisenberg model and define the associated 1:2 symplectic correspondence on
the domains $\mathfrak M_{\pm 1}$ of the cotangent bundle of pseudospheres $S^{n-1}_{\pm 1}$ (Theorem \ref{peta}),
i.e, the domain $\mathfrak M_0$ of the cotangent bundle of a light--like cone $S^{n-1}_0$ (Theorem \ref{osma}). The small
$2\times 2$--matrix representations for the systems are also given
(Theorems \ref{sedma}, \ref{deveta}).

We show that the Heisenberg model on
$\mathfrak M_{\pm 1}$ is a B\"acklund transformation (Theorem \ref{sedma}) of the
integrable variant of the Neumann system in pseudo--Euclidean
spaces described in Theorem \ref{cetvrta}.
On the other hand, the Heisenberg
model on $\mathfrak M_0$ has a one-parameter family of invariant
hypersurfaces $\Sigma_\kappa$. The restriction
of the correspondence  to $\Sigma_\kappa$ is a natural example of
completely integrable contact system (Theorem \ref{deseta}).

Motivated by Moser's geometric interpretation of the integrals of
the Neumann system on a sphere (see \cite{Mum}),
in section 5 we consider
the following
pseudo--confocal family of quadrics in $\mathbb E^{k,l}$
\begin{equation}\label{confocal}
\mathcal Q_{c,\lambda}\colon \quad  \langle (U-\lambda \mathbf I)^{-1} x,
x\rangle =\sum_{i=1}^n\frac{\tau_i x_i^2}{U_i-\lambda}=c, \quad \lambda
\ne U_i, \quad i=1,\dots,n,
\end{equation}
where $U_i=J^2_i$, $i=1,\dots,n$.
In the light--like case, to a given trajectory $\{q_k\,\vert\,k\in\mathbb Z\}$ we associate a sequence of planes
\[
\pi_k=\Span\{q_k,Jq_k\}, \qquad k\in\mathbb Z.
\]
Then, if $\pi_k$ is tangent to
a cone $\mathcal Q_{0,\lambda^*}$ from the pseudo--confocal
family \eqref{confocal} for a certain $k$,  then it is tangent to $\mathcal
Q_{0,\lambda^*}$ for all $k\in\mathbb Z$. In the case $c=\pm 1$, instead of planes, to a trajectory $\{q_k\,\vert\,k\in\mathbb Z\}$ we associate sequence of lines
\[
l_k=Jq_k+\Span\{q_k\}, \qquad k\in\mathbb Z
\]
with the same property (Theorem \ref{jedanaesta}). Further, under the condition
$
U_1<U_2<\dots<U_n,
$
we estimate the number of (real) quadrics tangent to planes $\pi_k$ (lines $l_k$) for a generic trajectory $\{q_k\,\vert\,k\in\mathbb Z\}$ (Theorem \ref{trinaesta}).

\section{Heisenberg model and billiards}

\subsection{Generating functions for systems with constraints}
In what follows, we will use the following simple observation (see \cite{Jov3}).
Consider $(2n-2m)$--dimensional submanifolds $M\subset \R^{2n}(x,p)$ and $N \subset \R^{2n}(X,P)$, defined
by the constraints of the form
\begin{align*}\label{constraints2}
&M\colon \qquad f_i(x)=0, \qquad \, \, f_{m+i}(p,x)=0, \,\,\qquad i=1,\dots,m,\\
&N\colon \qquad F_i(X)=0, \qquad F_{m+i}(P,X)=0, \qquad i=1,\dots,m.
\end{align*}
We suppose that $M$ and $N$ are symplectic submanifolds with respect to the canonical symplectic forms, that is
\[
\det(\{f_i,f_j\})\ne 0\vert_M, \qquad
\det(\{F_i,F_j\}) \ne 0\vert_N, \qquad i,j=1,\dots,2m,
\]
where $\{\cdot,\cdot\}$ are the canonical Poisson bracket (e.g., see \cite{Su}).

\begin{thm}\label{prva}
If a graph $\Gamma_\phi$ of the diffeomorphism $\phi\colon M\to N$ can be given by
\begin{equation}\label{GF2}
p=\frac{\partial S(x,X)}{\partial x}+\sum_{i=1}^m \lambda_i \frac{\partial f_i}{\partial x}, \qquad
P=-\frac{\partial S(x,X)}{\partial X}-\sum_{i=1}^m \Lambda_i \frac{\partial F_i}{\partial X},
\end{equation}
for certain Lagrange multipliers $\lambda_i, \Lambda_i$,
then $\phi$ is symplectic. Similarly, if \eqref{GF2} defines a diffeomorphism $\phi\colon M\to N$,
then $\phi$ is symplectic.
\end{thm}

\subsection{Virtual billiards}
Let
\begin{equation}\label{elipsoid}
\mathbb{Q}^{n-1}=\left\{x\in \mathbb E^{k,l}\,|\, \langle A^{-1}x,x\rangle =c \right\}, \qquad c=\pm 1, 0
\end{equation}
be a $(n-1)$--dimensional quadric,
where
\[
A=\diag(a_1,\dots,a_n), \qquad \det A\ne 0.
\]
A point $x\in\mathbb{Q}^{n-1}$ is {\it singular} if the induced metric is degenerate at
$x$, i.e., if a pseudo--Euclidean normal $\A x$ at $x$
is light--like: $\langle A^{-2}x,x\rangle=0$.\footnote{The matrix $A$ used here, corresponds to the matrix $EA$ used in \cite{JJ3}.}

The \emph{virtual billiard mapping} $\phi\colon (x_k,y_k)\mapsto (x_{k+1},y_{k+1})$ is defined by:
\begin{align}
x_{k+1}&=x_k+\mu_k y_k=x_k-2\frac{\langle\A x_k,y_k\rangle}{\langle \A y_k,y_k\rangle }y_k,\label{1bilijar}\\
y_{k+1}&=y_k+\nu_k A^{-1}x_{k+1}=y_k+2\frac{\langle\A
x_{k+1},y_{k+1}\rangle}{\langle A^{-2} x_{k+1},x_{k+1}\rangle }A^{-1}x_{k+1}, \label{2bilijar}
\end{align}
where the multipliers
$\mu_k$, $\nu_k$
are determined from the conditions that the "impact" points $x_k$ belong to the quadric \eqref{elipsoid}
 and that the outgoing  and incoming directions at $x_{j+1}$ have the same norms: $ \langle y_{k+1},y_{k+1}\rangle=\langle y_k,y_k\rangle$.

Geometrically \eqref{2bilijar} means that $y_{k} \mapsto y_{k+1}$ is the billiard reflection at $x_{k+1}\in \mathbb Q^{n-1}$ in the pseudo--Euclidean space $\mathbb E^{k,l}$, but $\mu_k$ in \eqref{1bilijar} can be less then zero as well. Thus, the segments
$x_{k-1}x_k$ and $x_k x_{k+1}$ determined by 3 successive points
of the mapping \eqref{1bilijar}, \eqref{2bilijar} may be
either on the same side of the tangent plane $T_{x_k}\mathbb
Q^{n-1}$ (the usual billiard reflection at $x_k$), or
 on the opposite sides of
$T_{x_k}\mathbb Q^{n-1}$. Such configurations were studied in \cite{DR2006, DrRa, DR2012, Gl}.

The system is defined
outside the singular set
\begin{equation}\label{singular}
\Sigma=\{(x,y)\in \R^{2n}\,\,\vert\,\,\langle A^{-2}x,x\rangle =0 \,\, \vee
\langle A^{-1}x,y\rangle =0\,\,\vee\,\,\langle A^{-1} y,y\rangle =0\},
\end{equation}
and if $(x_{k+1},y_{k+1})$ is singular, the flow stops.
The lines
$l_k=\{x_k+sy_k\,\vert\,s\in\R\}$ containing segments $x_kx_{k+1}$
of a given virtual billiard trajectory are of the same type: they are all
either space--like ($\langle y_k,y_k\rangle>0$), time--like ($\langle y_k,y_k\rangle<0$) or light--like
($\langle y_k,y_k\rangle=0$). Also, the function $\langle \A x_k,y_k\rangle$ is the first integral of
the system.

Consider the submanifold of the symplectic linear space $\R^{2n}(x,y)$
\begin{align*}
M_{c,h} =\{(x,y)\in
\R^{2n}\backslash\Sigma\,\,\vert\,\,\phi_1=\langle A^{-1}x,x\rangle =c, \,\,
\phi_2=\langle y,y\rangle=h\},
\end{align*}
where we take the symplectic form
\[
\sum_i \tau_i dy_i\wedge dx_i
\]
obtained from the canonical symplectic form on $\R^{2n}(x,p)$ after the identification $p=Ey$.
Since $\{\phi_1,\phi_2\}=4\langle A^{-1}x,y\rangle \ne 0\vert_{M_{c,h}}$, it follows
that $M_{c,h}$ is a symplectic submanifold of $\R^{2n}(x,y)$ and the mapping $\phi$ is a symplectic transformation of $M_{c,h}$
(see Theorem 2.1, \cite{JJ3}\footnote{In Theorem 2.1, \cite{JJ3} a direct proof in terms of the induced Dirac--Poison brackets on $M_{c,h}$ is given for $c=1$, but the same proof applies for $c=0$ and $c=-1$.}).\footnote{Here, as in the third footnote we note that for the codomain of $\phi$ we should consider the variety $M_{c,h}$ without the assumptions
$\langle A^{-2}x,x\rangle \ne 0,\langle A^{-1} y,y\rangle \ne 0$.}

The Hamiltonian and contact integrability of the virtual billiard dynamics
is described in \cite{JJ3}.
In the case when $EA$ is positive definite, $c=+1$, this is a billiard system within ellipsoid ${\mathbb Q}^{n-1}$ in the pseudo-Euclidean space (see  \cite{KT, DR}).

For $c=0$, the dynamics \eqref{1bilijar}, \eqref{2bilijar} induces
a well defined dynamics of the lines $\Span\{x_k\}$, i.e, the points  $p_k=[x_k]\in \mathbb Q^{n-2}$ of the $(n-1)$--dimensional projective space $\mathbb
P(\mathbb E^{k,l})$
outside the singular set
$$
\Xi=\{[x]\in \mathbb P(\mathbb E^{k,l})\,\,\vert\,\,
\langle A^{-2}x,x\rangle =0\},
$$
where $\mathbb Q^{n-2}$ is the projectivization of the cone
\eqref{elipsoid} within $\mathbb P(\mathbb E^{k,l})$.
A sequence $\{p_k\}$ is a
{billiard trajectory} within the quadric $\mathbb Q^{n-2}$ in
the projective space $\mathbb P(\mathbb E^{k,l})$ with respect to the
metric induced from the pseudo--Euclidean space $\mathbb E^{k,l}$.
In particular, for the signature
$(n,0)$ and the condition
\[
0<a_1,a_2,\dots,a_{n-2},a_{n-1} <-a_{n},
\]
and the signature
$(n-1,1)$ with the
condition
\begin{equation}\label{condA2}
, \qquad
0<a_1,a_2,\dots,a_{n-2},a_{n-1} < a_{n},
\end{equation}
we obtain ellipsoidal billiards on the sphere
and the Lobachevsky space, respectively (see \cite{JJ3, Jov3}).

\subsection{The skew hodograph mapping and quadratic generating functions.}

There is a remarkable relation between the ellipsoidal Euclidean
billiards and the Heisenberg system established by the use of the so called
skew hodograph mapping (see Moser and Veselov \cite{MV}).
In \cite{Jov3}, the skew-hodograph
mapping is interpreted as a symplectic transformation with a quadratic generating function
for a system with constraints.
Here, we shall give analogous mapping for virtual billiards, which also
include the symmetries of the system. Another construction, related to pluri-Lagrangian systems, that associate generating functions to the billiard system within ellipsoid is recently given in \cite{S2}.

For the Euclidean case when $\mathbb Q^{n-1}$ is an ellipsoid, we have the following characterisation of quadratic
generating functions.

\begin{lem}
\label{karakterizacija}
A quadratic generating function
$
S(x,X)=\langle B x,X\rangle,
$
$\det B\ne 0$, defines
a symplectic transformation $\psi\colon M_{1,1}\to M_{1,1}$ within a real domain only if $\vert B^T A^{1/2}\vert=\vert B A^{1/2}\vert = 1$, where $\vert\cdot\vert$ is the operator norm of the matrix.
\end{lem}

The proof of Lemma \ref{karakterizacija} is given in the Appendix.
Apart from the obvious solution $B={A}^{-1/2}$ of the stated necessary conditions
that leads to the skew-hodograph mapping (see \cite{Jov3}), we have a family of solutions related to the symmetry of the
ellipsoid $\mathbb{Q}^{n-1}$. Namely, let $\mathbf R\in O(n)$ be an orthogonal matrix that commute with $A \colon \Ad_\mathbf R(A)=A$.
Then we can take
$
B=\mathbf RA^{-1/2}=A^{-1/2}\mathbf R.
$

The above construction can be considered in pseudo--Euclidean spaces as well, provided $A$ is positive definite.
Recall that if some of the eigenvalues of the matrix $A$ are the same, we deal
with virtual billiards with symmetries and the corresponding dynamics is integrable in a noncommutative sense (see \cite{JJ3}).
The set of all symmetries
\[
\mathbf R\in O(k,l)\colon  \quad \Ad_\mathbf R(A)=A
\]
($\mathbf RE\mathbf R^T=E$, $\mathbf R A \mathbf R^{-1}=A$)
is isomorphic to
$
O(k_1,l_1)\times\dots\times O(k_r,l_r).
$
If all eigenvalues of $A$ are distinct, the only symmetries are
\[
\mathbf R\in \mathbb Z_2^n\subset O(k,l), \quad \text{i.e}, \quad
\mathbf R=\diag(\pm 1,\dots, \pm 1),
\]
and the system is integrable in the usual commutative sense.

Let
\[
B=\mathbf RA^{-1/2},
\]
where $\mathbf R\in O(k,l)$ is a symmetry of the quadric.
Consider the symplectic manifold $M_{c,c}$, $c=\pm 1, 0$ and the generating function
\[
S(x,X)=\langle Bx,X\rangle=\langle \mathbf RA^{-1/2} x,X\rangle.
\]

The equations \eqref{GF2} become
\begin{align}
Ey&=B^TE X+\lambda E\A x=\mathbf R^TA^{-1/2} EX+\lambda E\A x,\label{1nova++}\\
EY&=-EBx-\Lambda E\A X=-E\mathbf RA^{-1/2} x-\Lambda E\A X, \label{2nova++}
\end{align}
where $\langle \A x,x\rangle =c$, $\langle \A X,X\rangle=c$. We have four real values for $(\lambda,\Lambda)$ given by
\begin{align*}
&\lambda=0 \quad {\text{or}} \quad \lambda=-2\langle \A x,\mathbf R^{-1}A^{-1/2} X\rangle /\langle  A^{-2}x,x\rangle,\\
&\Lambda=0 \quad {\text{or}} \quad \Lambda=-2 \langle \A X,\mathbf  RA^{-1/2} x\rangle  /\langle A^{-2}X,X\rangle.
\end{align*}

For
$\lambda=0$, $\Lambda\ne 0$,
the relations \eqref{1nova++}, \eqref{2nova++} define the symplectic mapping $\psi_\mathbf R\colon M_{c,c}\to M_{c,c}$
given by
\begin{align}
X&=\mathbf RA^{1/2}y,\label{1nova**}\\
Y&=-\mathbf RA^{-1/2}(x+\mu y), \qquad \mu=-2\langle \A x,y\rangle/\langle \A y,y\rangle.\label{2nova**}
\end{align}

Let $\mathbf I$ be the identity $n\times n$--matrix.

\begin{thm}\label{druga}
{\rm (i)} The mapping $\psi_\mathbf R$ commute with
the virtual billiard mapping $\phi$. In other words, let $(x_k,y_k)$ be a solution of
\eqref{1bilijar}, \eqref{2bilijar} with $\langle
y_k,y_k\rangle=c$.
Then $(x_k',y_k')=\psi_\mathbf R(x_k,y_k)$, is a solution of \eqref{1bilijar},
\eqref{2bilijar} with $\langle y'_k,y'_k\rangle=c$.
Moreover,  $\psi^2_\mathbf R=-\mathbf R^2\circ \phi:$
$$
x_k''=-\mathbf R^2 x_{k+1}, \qquad y_{k}''=-\mathbf R^2 y_{k+1}.
$$

{\rm (ii)}  Let $(x_k,y_k)$ be a trajectory of the mapping $\psi=\psi_\mathbf{I}$.
Then $q_k=y_k$ is a
solution of the Heisenberg model \eqref{neumann} on $S^{n-1}_{c}$ with $J=A^{1/2}$. Conversely, if $J$ is
positive definite and $q_k$ is a solution of the Heisenberg system
\eqref{neumann} on $S^{n-1}_{c}$, then
\[
x_k=(-1)^k J q_{2k}, \qquad \tilde{x}_{k}=(-1)^k J q_{2k+1}
\]
are billiard trajectories within the quadric $\langle A^{-1}x_k, x_k\rangle=c$,
where $A=J^2$.
\end{thm}

\begin{proof} (i) Let $(x_k,y_k)$ be a
solution of \eqref{1bilijar}, \eqref{2bilijar} with $\langle
y_k,y_k\rangle=c$ and let
$$
(x_k',y_k')=\psi_\mathbf R (x_k,y_k)=(\mathbf RA^{1/2}y_k,-\mathbf RA^{-1/2}x_{k+1}).
$$
Then with $k$ replaced by
$k+1$  we obtain, respectively,
\begin{align*}
& x_{k+1}'=x_k'+\mu_k' y_k',            \\
& y_{k+1}'=y_k'+\nu_k' A^{-1} x_{k+1}',
\end{align*}
where $\mu_k'=-\nu_k$, $\nu_k'=-\mu_{k+1}$.

Further, we have
\begin{align*}
x_k'' &=\mathbf R A^{1/2} y_k'=- \mathbf R A^{1/2} \mathbf R A^{-1/2} x_{k+1}=-\mathbf R^2 x_{k+1},\\
y_k'' &=-\mathbf R A^{-1/2}x_{k+1}'=-\mathbf R A^{-1/2}\mathbf R A^{1/2}y_{k+1}=-\mathbf R^2 y_{k+1}.
\end{align*}

(ii) The second statement follows from the relations
\begin{align*}
y_{k+2} &=-A^{-1/2}(x_{k+1}+\mu_{k+1}y_{k+1})\\
        &=-A^{-1/2}(A^{1/2}y_k+\mu_{k+1}y_{k+1})\\
        &=-y_k-\mu_{k+1}A^{-1/2}y_{k+1}. \qedhere
\end{align*}
\end{proof}

For $\psi=\psi_\mathbf I$, we have the following commutative diagram
\begin{equation}\label{comDiagram}
\xymatrix@R45pt@C60pt{
P_c \ar@{^{}->}[r]^{\Delta} \ar@{^{}->}[d]^{\Phi} & M_{c,c} \ar[r]^{\phi} \ar@{^{}->}[d]^{\psi}& M_{c,c} \ar@{^{}->}[d]^{\psi} \\
P_c \ar@{^{}->}[r]^{\Delta}  & M_{c,c}  \ar[r]^{\phi}& M_{c,c} }
\end{equation}
where $\Delta \colon P_c\to M_{c,c}$ is a symplectomorphism
 $x=Jq, y=Q$, $J=A^{1/2}$.

Also, since $\psi^2=-\phi$, if $q_k$ is periodic with period
$4N$ (respectively, $4N+1, 4N+2, 4N+3$), then $x_k$, $\tilde x_k$ are periodic
with period $2N$ (respectively, $8N+2, 4N+2, 8N+6$).

\begin{figure}[ht]
\includegraphics[width=80mm]{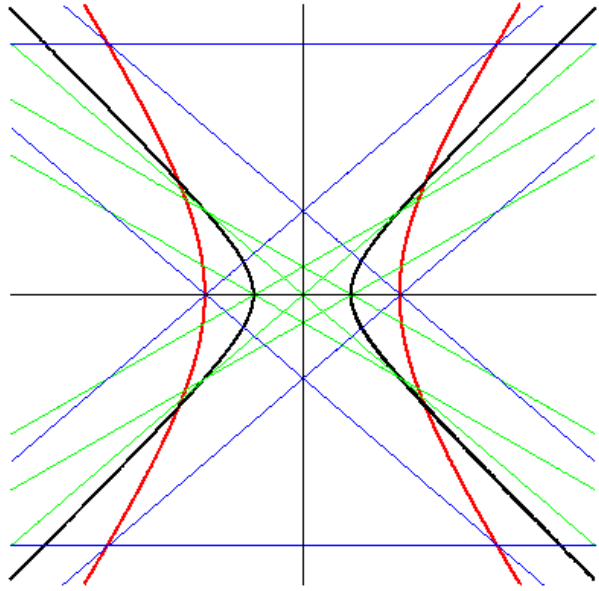}
\caption{6--periodic trajectory of the Heisenberg model (green
lines) and the corresponding 6--periodic space--like trajectory of
the virtual billiard for $c=1$ (blue lines) in $E^{1,1}$.}
\end{figure}



In the signature $(n-1,1)$ the statement relates the
ellipsoidal billiard on the Lobachevsky space and the Heisenberg system on the light--like cone $S^{n-1}_0$ with the matrix $A$ given by \eqref{condA2} (see \cite{Jov3}).

\begin{exm}\label{ST0}{\rm
As an example of a system with symmetry, consider
the billiard within ellipsoid
\[
\mathbb Q^{2n-1}\colon \quad \langle A^{-1} \mathbf z,\bar{\mathbf z}\rangle=\frac{\vert z_1\vert^2}{a_1}+\dots+\frac{\vert z_n\vert^2}{a_n}=1
\]
in the Euclidean space $\R^{2n}\cong \mathbb C^{n}$, $\mathbf z=(z_1,\dots,z_n)$
(in \cite{BozaRos} we studied the reduction of symmetries of the given billiard with additional Hook's potential).  In the complex notation we have
\[
M_{1,1}=\{(\mathbf z,\mathbf v)\in\mathbb C^{2n} \,\vert\, \langle A^{-1} \mathbf z,\bar{\mathbf z}\rangle=1, \langle \mathbf v,\bar{\mathbf v}\rangle=1\}.
\]

Note that for $a_i\ne a_j$, $i\ne j$, we have $O(2)^n$ (i.e, $U(n)^n$) symmetry of $\mathbb Q^{2n-1}$
and the symplectic mapping \eqref{1nova**}, \eqref{2nova**} reads
\begin{align}
\mathbf z_{k+1}&=\mathbf RA^{1/2}\mathbf v_k,\label{ST1}\\
\mathbf v_{k+1}&=-\mathbf RA^{-1/2}(\mathbf z_k+\mu_k \mathbf v_k), \quad \mu_k=-2{\Re}\langle \A z_k,\bar v_k\rangle/\langle \A v_k,\bar v_k\rangle,
\label{ST2}
\end{align}
where
\[
\mathbf R=(e^{i\theta_1},\dots,e^{i\theta_n}).
\]

In particular, for $\mathbf R=(i,\dots,i)=iE$ we obtain $\psi^2_{iE}=\phi$, that is \eqref{ST1}, \eqref{ST2} is exactly the square root of the billiard. This mapping coincides with the symplectic billiard mapping for the ellipsoid $\mathbb Q^{2n-1}$ introduced in \cite{AT}. More precisely, after setting $k+1$ instead $k$ in \eqref{ST1},  we get
\[
\mathbf z_{k+2}= i A^{1/2}\mathbf v_{k+1}= i A^{1/2}(-i A^{1/2})(\mathbf z_k+\mu_k \mathbf v_k)= \mathbf z_k-\mu_k i A^{-1/2}\mathbf z_{k+1}.
\]
Thus,  $\{\mathbf z_k\}$ is a trajectory of
a symplectic billiard within the ellipsoid $\mathbb Q^{2n-1}$ (corresponding to the ellipsoid (18) in \cite{AT}, where
we set $a_i$ instead of $a_i^2$).
}\end{exm}

\subsection{The $(n\times n)$--matrix representation of the virtual billiards.}

Motivated by the Lax representation for elliptical billiards with the Hooke's potential (Fedorov \cite{fedo}, see also \cite{R}),
a "small" $2\times 2$ matrix representation for the virtual billiard mapping is given in \cite{JJ}.
On the other hand, in \cite{Jov2} we presented the following "big" $n\times n$--matrix representation of the Heisenberg system, a modification of the $n\times n$ matrix
representation given in \cite{MV}.
Let
\[
F=\diag(1,\dots,1,i,\dots, i),
\]
where the first $k$ components are equal to 1, and the last $n-k$
components are equal to the imaginary unit $i$ ($F^2=E$). The
equations \eqref{neumann} imply the discrete Lax representation
\begin{equation}\label{velikiLA}
L_{k+1}(\lambda)=A_k(\lambda)L_k(\lambda)A_k^{-1}(\lambda),
\end{equation}
where
\begin{eqnarray*}
&& L_{k}(\lambda)=J^2+\lambda Fq_{k-1}\wedge FJ q_k-c\cdot\lambda^2 Fq_{k-1}\otimes Fq_{k-1},\\
&& A_k(\lambda)=J-\lambda Fq_k\otimes Fq_{k-1}.
\end{eqnarray*}

Note that in the light--like cone case, the $L$--matrix is linear
in $\lambda$. Also, if $J^2_j\ne J^2_i$, the integrals obtained from the
matrix representation can be written in the form \eqref{int-neumann}.

In the Euclidean case, the skew hodograph mapping relates $n\times n$ matrix representations of
the Heisenberg model and the elliptic
billiard \cite{MV}. Although we have
an analog of the skew hodograph mapping only for $A>0$, the
modification of the matrix representation for the Heisenberg model
from the Euclidean to the pseudo--Euclidean spaces \eqref{velikiLA} suggests the following matrix representation for
virtual billiards.

\begin{thm}\label{treca}
The virtual billiard mapping \eqref{1bilijar}, \eqref{2bilijar} implies the discrete Lax
representation
$$
\mathcal L_{k+1}(\lambda)=\mathcal A_k(\lambda)\mathcal
L_k(\lambda)\mathcal A_k^{-1}(\lambda),
$$
where
\begin{align*}
& \mathcal L_{k}(\lambda)=A-\lambda Fx_{k}\wedge F
y_{k-1}- c \cdot \lambda^2
Fy_{k-1}\otimes Fy_{k-1},\\
&\mathcal  A_k(\lambda)=A-\lambda (Fx_k\otimes
Fy_{k-1}-Fy_k\otimes Fx_k)- c\cdot  \lambda^2  Fy_{k}\otimes Fy_{k-1}.
\end{align*}
\end{thm}

The proof is given in the Appendix.

\section{B\"acklund transformation of the Neumann system}

\subsection{Continuous limit and the Neumann system.}

Following
Moser and Veselov \cite{MV}, by taking $J(\epsilon)=\mathbf I+\frac12
\epsilon^2 U$, $q_k=q(t_0+k\epsilon)$ for small $\epsilon$, from
\eqref{neumann} we obtain the equation for $q(t)$
\begin{equation}\label{aproksimacija}
(\mathbf I+\frac12 \epsilon^2 U)(2q+\epsilon^2\ddot{q})\approx
\lambda q,
\end{equation}
that is
\[
\ddot{q}\approx -Uq+ \mu q, \qquad \mu=(\lambda-2)\epsilon^{-2}.
\]

The last equation in the case $c=\pm 1$ describes {\it the
Neumann system on a pseudosphere}. This is the
Lagrangian system with the Lagrangian
$$
L(q,\dot q)=\frac12\langle \dot q,\dot q\rangle -\frac12\langle
Uq,q\rangle,
$$
subjected to the constraint $\langle q,q\rangle=c$, $c=\pm 1$,
where
\[
U=\diag(U_1,\dots,U_n).
\]
Indeed, the associated Euler--Lagrange
equation on the tangent bundle $TS^{n-1}_c$ realized by equations
\begin{equation}\label{tS}
\langle q,q\rangle=c, \qquad \langle q,\dot q\rangle=0,
\end{equation}
reads
\begin{equation}\label{c-neumann}
\ddot{q}=-Uq+\mu q, \qquad \mu=-\frac{1}c\big(\langle \dot q,\dot q \rangle-\langle
Uq,q\rangle\big).
\end{equation}
In the case of the light--like cone $c=0$, the Lagrangian $L$ is
degenerate, since all points of $S^{n-1}_0$ are singular.

We will show that the cotangent bundle formulation of the
Heisenberg model provides a  B\"acklund transformation of the
Neumann system. The construction is a straightforward generalization
of the discretization of the Neumann system presented by Suris \cite{Su}.

Firstly, we need a Hamiltonian formulation of the
Neumann flow.
Consider the realization of the cotangent
bundle $T^*S^{n-1}_c$ as a submanifold of $\R^{2n}(q,p)$ endowed with the canonical symplectic form $\omega=dp\wedge dq$:
\begin{equation}\label{pseudosphere}
T^*S^{n-1}_c\colon \quad \varphi_1=\langle q,q\rangle=c, \quad
\varphi_2=\langle q,Ep\rangle=0.
\end{equation}

This is a symplectic submanifold for $c=\pm 1$, since
$\{\varphi_1,\varphi_2\}=2\langle q,q\rangle=c\ne 0$ at
$T^*S^{n-1}_c$. Moreover, the restriction
$\omega\vert_{T^*S^{n-1}_c}$ coincides with the canonical
symplectic form on $T^*S^{n-1}_c$. The induced Poisson--Dirac
bracket reads
\begin{equation}
\{f_1,f_2\}_D=\{f_1,f_2\} -\frac{
\{\varphi_1,f_1\}\{\varphi_2,f_2\}-\{\varphi_2,f_1\}\{\varphi_1,f_2\}}{\{\varphi_1,\varphi_2\}}.
\label{Dirac_bracket2}
\end{equation}

The equation
\begin{equation}\label{LT}
p=\frac{\partial L}{\partial \dot q}+\lambda Eq=E\dot q+\lambda Eq
\end{equation}
with constraints \eqref{tS} and \eqref{pseudosphere} implies $\lambda=0$ and $\dot q=Ep$.
Thus, the Legendre transformation of $L(q,\dot q)$
yields the Hamiltonian function
\begin{equation}\label{Ham}
H(q,p)=\frac12\langle p,p\rangle +\frac12\langle Uq,q\rangle.
\end{equation}

The equations \eqref{c-neumann} are equivalent to the Hamiltonian
equations with constraints
\begin{eqnarray}
&& \dot q=\frac{\partial H}{\partial p}-\mu\frac{\partial \varphi_1}{\partial p}-\nu\frac{\partial \varphi_2}{\partial p}=Ep-\nu q,  \label{ham1} \\
&& \dot p=-\frac{\partial H}{\partial
q}+\mu\frac{\partial \varphi_1}{\partial q}+\nu\frac{\partial
\varphi_2}{\partial q}=-EUq+\mu Eq+\nu p, \label{ham2}
\end{eqnarray}
where the multipliers $\mu,\nu$, determined from the conditions
$\dot \varphi_1=\dot\varphi_2=0$ are given by
\[
\mu=-\frac{1}{\langle q,q\rangle}\left(\langle p,p \rangle-\langle
Uq,q\rangle\right), \qquad \nu=0.
\]

Let
\begin{equation}\label{QL}
Q_\lambda(x,y)=\langle (\lambda\mathbf I-U)^{-1} x,y\rangle=\sum_i
\frac{\tau_i x_i y_i}{\lambda-U_i}.
\end{equation}

\begin{thm}\label{cetvrta}
The Neumann flow \eqref{ham1}, \eqref{ham2} implies the matrix
representation
\begin{equation}\label{neumannLA}
\frac{d}{dt}\mathcal{L}_{q,p}(\lambda)=[\mathcal{L}_{q,p}(\lambda),\mathcal{A}_{q,p}(\lambda)],
\end{equation}
with $2\times2$ matrices depending on the parameter $\lambda$
\begin{align*}
\mathcal{L}_{q,p}(\lambda)&=\left(\begin{array}{cc}
-Q_{\lambda}(q,Ep) & -Q_{\lambda}(q,q)\\
c+Q_{\lambda}(p,p) & Q_{\lambda}(q,Ep)
\end{array}\right),  \quad
\mathcal{A}_{q,p}(\lambda)= \left(\begin{array}{cc}
0 & 1 \\
\mu-\lambda & 0
\end{array}\right).
\end{align*}

The system is completely integrable.
For $U_i\ne U_j$, $i\ne j$, from the expression
\begin{equation}\label{DETc}
\det \mathcal
L_{q,p}(\lambda)=Q_\lambda(q,q)(c+Q_\lambda(p,p))-Q_\lambda(q,Ep)^2=\sum_{i=1}^n\frac{f_i(q,p)}{\lambda-U_i},
\end{equation}
we obtain a complete set of integrals
\begin{equation}
\label{int-c} f_i(q,p)=c\cdot \tau_i q_i^2+\sum_{j\neq
i}\frac{\tau_i\tau_j(\tau_jp_jq_i-\tau_iq_jp_i)^2}{U_i- U_j},
\end{equation}
where $\{f_i,f_j\}_D=0$, $i,j=1,\dots,n$, and $\sum_i f_i \equiv 1$.
\end{thm}

A "big" $n\times n$ matrix representation and integration of equations \eqref{c-neumann} in the signature $(n-1,1)$, i.e,
of the Neumann system in the Lobachevsky space is given by Veselov (see Appendix B, \cite{Ves2}).
A generalization of the Neumann system to the Stiefel varieties, as well as its
integrable discretization, is given in \cite{FJ} and \cite{FeJo}, respectively.

\subsection{Discrete Legendre transformation for $c=\pm 1$.}

Following \cite{MV, Ves3, Su}, we consider the associated discrete cotangent bundle dynamics of the Heisenberg system.
Let  $\mathfrak M_c$, $c=\pm 1$, be a domain
 within $T^*S^{n-1}_c$ defined by the inequalities
\begin{equation}\label{nejednakost}
D_c(q,p)=\langle J^{-2}E p,q\rangle^2 -\langle J^{-2} q,q\rangle(\langle J^{-2}p,p\rangle-c)> 0, \quad \langle J^{-2}q,q\rangle\ne 0.
\end{equation}

\begin{thm}\label{peta}
The relations $\Psi$ defined by
\begin{align}
Q=& EJ^{-1}p+\mu J^{-1}q,\label{Q}\\
P=& -EJq+\mu EQ,\label{P}
\end{align}
where $\mu$ is the solution of the quadratic equation
\begin{equation}\label{kvadratna}
\langle J^{-2} q,q\rangle \mu^2+2\mu \langle J^{-2}E p,q\rangle+\langle J^{-2}p,p\rangle-c=0,
\end{equation}
define $1:2$ symplectic correspondence $\Psi\colon \mathfrak M_c\to\mathfrak M_c$ ($c=\pm 1$).
\end{thm}

\begin{proof}
Consider a transformation of $T^*S^{n-1}_c$ defined by constraints \eqref{pseudosphere}
 and the generating function given by the discrete Lagrangian:
\[
 S(q,Q)=\langle JQ,q\rangle.
\]
The equations \eqref{GF2} read
\begin{eqnarray}
&&p=EJQ+\lambda Eq, \label{L1}\\
&&P=-EJq-\Lambda EQ, \label{L2}
\end{eqnarray}
where the Lagrange multipliers
$
\lambda=\Lambda=-\langle JQ,q\rangle/c
$
are determined from the constraints
$\langle Eq,p\rangle =\langle EQ,P\rangle=0$.

Let
$\mathbb L_1\colon P_c(q,Q) \to T^* S^{n-1}(q,p)$, $\mathbb L_2\colon  P_c(q,Q) \to T^* S^{n-1}(Q,P)$
be the mappings defined by \eqref{L1} and \eqref{L2}, respectively.
They can can be seen as a discrete analogue of the Legendre transformation
\eqref{LT}. Let
\[
\mathfrak N_c=\mathbb L_1(P_c)=\mathbb L_2(P_c).
\]

We have that $D_c(q,p)$ is greater then zero on $\mathfrak N_c$:
\begin{align*}
D_c & =\langle J^{-2}E p,q\rangle^2 -\langle J^{-2} q,q\rangle(\langle J^{-2}p,p\rangle-c)\\
& =\langle J^{-2}q,JQ+\lambda q\rangle^2 -\langle J^{-2} q,q\rangle(\langle J^{-1}Q+\lambda J^{-2}q,JQ+\lambda q\rangle-c)\\
&=\langle J^{-1}q,Q\rangle^2>0,
\end{align*}
since $\langle J^{-1}q,Q\rangle\ne 0$ at $P_c$.
Thus, $\mathfrak N_c$ is a subset of $\mathfrak M_c$.

Vice versa, assume $(q,p)\in\mathfrak M_c$. The relation \eqref{L1}, can be rewritten into the form \eqref{Q}, where $\mu$ is unknown multiplier. From the constraint $\langle Q,Q\rangle=c$ we get the equation \eqref{kvadratna} determining $\mu$ as a two-valued function of $(q,p)$
\[
\mu(q,p)=\frac{-\langle J^{-2}E p,q\rangle\pm \sqrt{D_c(q,p)}}{\langle J^{-2} q,q\rangle}.
\]
As a result we obtain two points $Q_1$ and $Q_2$ such that $(q,p)=\mathbb L_1(q,Q_1)=\mathbb L_1(q,Q_2)$, and $\mathfrak M_c=\mathfrak N_c$.

Therefore, according to Theorem \ref{prva}, we get a two-valued
symplectic transformation
\[
\Psi\colon \mathfrak M_c(q,p)\to
\mathfrak M_c(Q,P)
\]
such that $\Psi(q,p)=\mathbb L_2(\mathbb L_1^{-1}(q,p))$.
\end{proof}

Since all equations are algebraic, we have that \eqref{P}, \eqref{Q} is a symplectic 1:2 correspondence on $T^*S^{n-1}_c$ for complexified objects
with $D_c=0$ defining the set of branch points.
Note that the discriminant
$4D_c$ is the first integral of \eqref{P}, \eqref{Q}.
It can be verified directly. Also it follows from the
Lax representation \eqref{disr_Lax_N} given below. Namely,
\[
D_c=-\det\mathcal L_k\vert_{\lambda=-\epsilon^{-2}} .
\]

Recall that the commutative diagram \eqref{comDiagram} relates the Heisenberg system with the virtual billiard dynamics.
Now we have:

\begin{lem}\label{komutativna}
The following diagram is commutative
\begin{equation*}
\xymatrix@R45pt@C60pt{
P_c(q,Q) \ar@{^{}->}[r]^{\mathbb L_1} \ar[dr]^{\mathbb L_2} \ar[d]_{\Phi} & \mathfrak M_c(q,p) \ar[d]^{\Psi} \\
P_c(q,Q)  \ar@{^{}->}[r]^{\mathbb L_1}                              & \mathfrak M_c(Q,P) }
\end{equation*}
in the sense that two-valued map $\Psi$ satisfies $\Psi(q,p)=\mathbb L_1\big(\Phi\big(\mathbb L_1^{-1}(q,p)\big)\big)$.
\end{lem}

Lemma \ref{komutativna} is a direct corollary of the definition of discrete Legendre transformations \eqref{L1}, \eqref{L2} and the equation of the stationary
configuration \eqref{neumann}. For the completeness of the exposition the proof is included in the Appendix.

As a result, for $c=\pm 1$, we refer to the correspondence
\begin{align}
& q_{k+1}=EJ^{-1}p_k +\mu_k J^{-1} q_k,\label{n1}\\
& p_{k+1} =  -EJq_{k}+\mu_k Eq_{k+1},  \qquad k\in\mathbb
Z,\label{n2}\\
& \mu_k=\big(-\langle J^{-2}E p_k,q_k\rangle\pm \sqrt{D_c(q_k,p_k)}\big)/{\langle J^{-2} q_k,q_k\rangle}
\nonumber
\end{align}
as the {\it Heisenberg model} on
$\mathfrak M_c$. If $\langle J^{-2} q_{k+1},q_{k+1}\rangle=0$, by definition the flow stops.

By subtracting \eqref{n1} and \eqref{n2}, where we set $k$ instead
of $k+1$, we obtain the equation of stationary configuration
\eqref{neumann} with the Lagrange multipliers \eqref{lambda} and $\mu_k$ related by
$$
\lambda_k=\mu_k+\mu_{k-1}, \qquad k\in\mathbb Z.
$$

By using the integrals \eqref{int-neumann}, Lemma \ref{komutativna} and the equation
\eqref{n1} we get:

\begin{thm}\label{sesta}
The Heisenberg system \eqref{n1}, \eqref{n2} is completely
integrable with a complete set of integrals
\begin{equation}\label{int-neumann2}
f_i(q,p)=c\cdot \tau_i q_i^2+\sum_{j\neq
i}\frac{\tau_i\tau_j(\tau_jp_jq_i-\tau_iq_jp_i)^2}{
J^2_i- J^2_j},
\end{equation}
where
$\{f_i,f_j\}_D=0$, $i,j=1,\dots,n$, and
 $\sum_i f_i \equiv 1$.
\end{thm}

\subsection{B\"acklund transformation}
Usually, a B\"acklund transformation for a system of differential
equations is a mapping which takes solution into solutions, or in
the framework of integrable systems, the symplectic mapping which
preserves Liouville folliation \cite{Su}. We saw that the
Moser--Veselov choice $J(\epsilon)=\mathbf I+\frac12\epsilon^2U$ for small
$\epsilon$ approximates the Neumann dynamics \eqref{c-neumann}.
However, it does not preserve the foliation given by
\eqref{int-c}. Instead, as in the Euclidean case (see
\cite{Su}), we take
$$
J(\epsilon)=\frac{1}{\epsilon}\sqrt{\mathbf I+\epsilon^2
U}=\frac{1}{\epsilon}\mathbf I+\frac12\epsilon U+\dots
$$
Then, from \eqref{neumann}, by taking $q_k=q(t_0+k\epsilon)$,
$\epsilon \approx 0$, we again obtain \eqref{aproksimacija} with
$\lambda$ replaced by $\epsilon\lambda$. Therefore, the Heisenberg
system with $J(\epsilon)=\frac{1}{\epsilon}\sqrt{\mathbf I+\epsilon^2 U}$ is
also a discretization of the Neumann system on the pseudosphere
$S^{n-1}_c$. On the other hand
$$
{1}/(J_i^2(\epsilon)-J_j^2(\epsilon))={1}/(U_i-U_j)
$$
and the integrals \eqref{int-neumann2} reduce to the integrals
\eqref{int-c}. Therefore, the corresponding Heisenberg model is a
B\"acklund transformation of the Neumann system. Moreover, we have
the following Lax representation depending on the parameter $\epsilon$.

\begin{thm}\label{sedma}
Let  $J^2(\epsilon)=U+\epsilon^{-2}\mathbf I$.
The Heisenberg system \eqref{n1}, \eqref{n2} on
$\mathfrak M_c$, $c=\pm 1$, implies the matrix equations with a
spectral parameter $\lambda$
\begin{gather} \label{disr_Lax_N}
\mathcal L_{k+1}(\lambda) = \mathcal M_k (\lambda) \mathcal
L_k(\lambda)\mathcal M_k (\lambda)^{-1},
\end{gather}
where
\begin{align*}
\mathcal{L}_{k}(\lambda)&=\left(\begin{array}{cc}
-Q_{\lambda}(q_k,Ep_k) & -Q_{\lambda}(q_k,q_k)\\
c+Q_{\lambda}(p_k,p_k) & Q_{\lambda}(q_k,Ep_k)
\end{array}\right),\\
\mathcal{M}_{k}(\lambda)&= \left(\begin{array}{cc}
-\mu_k & 1 \\
\mu_k^2-\lambda-\epsilon^{-2} & -\mu_k
\end{array}\right), \qquad c=\pm 1,
\end{align*}
and $Q_\lambda$ id given by \eqref{QL}.
\end{thm}

\section{Light--like cone and contact integrability.}

\subsection{Discrete Legendre transformation for the light--like case}
Instead of $\varphi_2$, for a description of the cotangent bundle
of the light--like cone $S^{n-1}_0$ we
use the function $\varphi_3=\langle p,q\rangle$:
\[
T^*S^{n-1}_0\colon\quad \varphi_1=\langle q,q\rangle=0, \quad
\varphi_3=\langle q,p\rangle=0.
\]

Then $\{\varphi_1,\varphi_3\}=2\langle Eq,q\rangle\ne
0$, for $q\ne 0$. Denote the new Dirac--Poisson bracket by
$\{\cdot,\cdot\}_D^0$.
Repeating the arguments from the previous section, by taking
$S(q,Q)=\langle q,J Q\rangle$ for a generating function,
we get the discrete Legendre transformations:
\begin{eqnarray*}
&&\mathbb L_1^0\colon P_0(q,Q)\to T^*S^{n-1}_0(q,p), \qquad\,\,\,\,  p=EJQ-  \frac{\langle EJ{Q},q\rangle}{\langle Eq,q\rangle} Eq,  \\
&&\mathbb L_2^0\colon P_0(q,Q)\to T^*S^{n-1}_0(Q,P), \qquad  P=-EJq +\frac{\langle EJQ,q\rangle}{\langle EQ,Q\rangle} EQ,
\end{eqnarray*}
and the 1:2
symplectic correspondence
\[
\Psi\colon \mathfrak M_0\to \mathfrak M_0
\]
given by
\begin{align*}
Q=& EJ^{-1}p+\mu J^{-1}q,\\
P=& -EJq+\tilde\mu EQ,
\end{align*}
where
\[
\mu=\big({-\langle J^{-2}E p,q\rangle\pm \sqrt{D_0(q,p)}}\big)/{\langle J^{-2} q,q\rangle}, \qquad \tilde\mu=\langle EJQ,q\rangle/\langle EQ,Q\rangle,
\]
 and $\mathfrak M_0$ is a subset of $T^*S^{n-1}_0$ defined by the inequalities
\eqref{nejednakost} for $c=0$.

Lemma \ref{komutativna} also applies, which together with
the integrals \eqref{int-neumann} implies the following statement.

\begin{thm}\label{osma}
The Heisenberg system on $\mathfrak M_0$
\begin{align}
q_{k+1} &=EJ^{-1}p_k+\mu_k J^{-1}q_k,\label{n01}\\
p_{k+1} &=  -EJq_{k}+\tilde\mu_k Eq_{k+1}, \qquad k\in\mathbb
Z,\label{n02}
\end{align}
is completely integrable. The complete set of first integrals is
\begin{equation}
f_i(q,p)=\sum_{j\neq
i}\frac{\tau_i\tau_j(\tau_jp_jq_i-\tau_iq_jp_i)^2}{J^2_i- J^2_j}, \label{INTEGRALI}
\end{equation}
where $\{f_i,f_j\}^0_D=0$, $i,j=1,\dots,n$, and
$\sum_i f_i \equiv 0$.
\end{thm}

Again, if $\langle J^{-2} q_{k+1},q_{k+1}\rangle=0$, by definition the flow stops.
Now, the Lagrange multipliers \eqref{lambda} of stationary
configuration \eqref{neumann} and the correspondence \eqref{n01}, \eqref{n02} are
related by
\[
\lambda_k=\mu_k+\tilde\mu_{k-1}, \qquad k\in\mathbb Z,
\]
and we have an analog of Theorem \ref{sedma}.

\begin{thm}\label{deveta}
Let $U=J^2$. The Heisenberg system \eqref{n01}, \eqref{n02} on
$\mathfrak M_0$ implies the matrix equations with a spectral
parameter $\lambda$
\begin{equation} \label{disr_Lax_N0}
\mathcal L_{k+1}(\lambda) = \mathcal M_k (\lambda) \mathcal
L_k(\lambda)\mathcal M_k (\lambda)^{-1},
\end{equation}
where
\begin{align*}
\mathcal{L}_{k}(\lambda)=\left(\begin{array}{cc}
-Q_{\lambda}(q_k,Ep_k) & -Q_{\lambda}(q_k,q_k)\\
Q_{\lambda}(p_k,p_k) & Q_{\lambda}(q_k,Ep_k)
\end{array}\right),\,\,
\mathcal{M}_{k}(\lambda)= \left(\begin{array}{cc}
-\mu_k & 1 \\
\mu_k {\tilde\mu}_k-\lambda & -\tilde{\mu}_k
\end{array}\right),
\end{align*}
and $Q_\lambda$ id given by \eqref{QL}.
\end{thm}

\begin{rem}{\rm
Obviously, the constraint $\varphi_3=0$ can be used for the Heisenberg
systems on $T^*S^{n-1}_{\pm 1}$ as well, but $\varphi_2=0$ is more
appropriate for a continuous Neumann system \eqref{c-neumann}. Namely,
the equation
\begin{equation*}\label{LT+}
p=\frac{\partial L}{\partial \dot q}+\lambda Eq=E\dot q+\lambda Eq
\end{equation*}
with constraints $\varphi_1=c$, $\varphi_3=0$ and \eqref{tS} implies $\lambda=-\langle E\dot q,q\rangle/\langle Eq,q\rangle$ and
\[
\dot q=Ep-\frac{\langle Ep,q\rangle}{\langle q,q\rangle}q.
\]

Thus, in this case, the Legendre transformation of $L(q,\dot q)$
yields the Hamiltonian function
\begin{equation*}\label{Ham+}
H(q,p)=\frac12\langle p,p\rangle -\frac12 \frac{\langle Ep,q\rangle}{\langle q,q\rangle}+\frac12\langle Uq,q\rangle,
\end{equation*}
having the extra term ${\langle Ep,q\rangle}/2{\langle q,q\rangle}$.
}\end{rem}

\subsection{Contact integrability}
The next statement is a cotangent variant of Theorems 2.1 and 3.3
given in \cite{Jov2}.

\begin{thm}\label{deseta}
{\rm (i)} The Heisenberg system \eqref{n01}, \eqref{n02} satisfies the invariant relation
$$
\langle E q_k,p_k\rangle +\langle E q_{k+1},p_{k+1}\rangle=0.
$$

{\rm (ii)}
The restriction of the correspondence
\eqref{n01}, \eqref{n02} to
the invariant manifold
$$
\Sigma_\kappa\subset \mathfrak M_0\colon  \qquad \varphi_2(q,p)=\langle E q,p\rangle=\pm\kappa, \qquad \kappa>0
$$
is a completely integrable
discrete contact system, with respect to the contact form $\theta=pdq\vert_{\Sigma_\kappa}$.
\end{thm}

\begin{proof} The statement follows from Theorems 2.1, 3.3 of \cite{Jov2} and Lemma \ref{komutativna}. For the completeness
of the exposition, we present a direct proof in the Appendix. \end{proof}

\section{Geometric interpretation of the integrals}

\begin{thm}\label{jedanaesta}
{\rm (i)} If a sequence of planes
\begin{equation}\label{PLANES}
\pi_j=\{ s_1 Ep_j+s_2 q_j\,\vert\, s_1,s_2\in\R\}, \qquad
j\in\mathbb Z
\end{equation}
determined by a trajectory $\{(q_j,p_j)\,\vert\,j\in\mathbb Z\}$
of the Heisenberg model \eqref{n01}, \eqref{n02} is tangent to a
cone $\mathcal Q_{0,\lambda^*}$ from the pseudo--confocal family
\eqref{confocal} for a certain $j$,  then it is tangent to
$\mathcal Q_{0,\lambda^*}$ for all $j\in\mathbb Z$.

{\rm (ii)} If a sequence of lines
\begin{equation}\label{LINES}
l_j=\{Ep_j+s q_j\,\vert,s\in\R\}, \qquad j\in\mathbb Z
\end{equation}
determined by a trajectory $\{(q_j,p_j)\,\vert\,j\in\mathbb Z\}$
of the Heisenberg model \eqref{n1}, \eqref{n2} is tangent to a
quadric $\mathcal Q_{c,\lambda^*}$ from the pseudo--confocal
family \eqref{confocal} for a certain $j$, then it is tangent to
$\mathcal Q_{c,\lambda^*}$ for all $j\in\mathbb Z$.
\end{thm}

\begin{proof} (i)
Let $\pi_I$, $I=(i_1,\dots, i_r)$, $1\le
i_1<i_2<\dots<i_r\le n$ be the Pl\"ucker coordinates of a
$r$--dimensional subspace $\pi$ in ${\mathbb R}^{n}$.
Then  $\pi$ is tangent to a nondegenerate cone $\mathcal K=\{\sum_i  b_i x_i^2 =0\}$  if and only if
$
\sum_I b_{i_1}\cdots
b_{i_k} \,\pi_I^2 =0
$ (see
Fedorov \cite{fe}).
For $r=2$, $\pi=\Span\{x,y\}$ the condition reduces to
\begin{equation}\label{f-uslov}
\big(\sum_i b_i x_i^2\big)\big(\sum_i b_iy_i^2\big)-\big(\sum_i b_i x_iy_i\big)^2=0.
\end{equation}

Thus, by taking $b_i=\tau_i/(U_i-\lambda^*)$, we get that
$\pi_j=\Span\{Ep_j,q_j\}$ is tangent to $\mathcal Q_{0,\lambda^*}$
if and only if
\begin{equation}\label{DET}
Q_{\lambda^*}(q_j,q_j)Q_{\lambda^*}(Ep_j,Ep_j)-Q_{\lambda^*}(q_j,Ep_j)^2=0.
\end{equation}
On the other hand, from Theorem \ref{deveta} we have that
\eqref{DET} is the integral of the system equal $\det\mathcal
L_j(\lambda^*)$. Therefore, if $\pi_j$ is tangent to $\mathcal
Q_{0,\lambda^*}$, it is tangent to $\mathcal Q_{0,\lambda^*}$ for
all $j\in\mathbb Z$.

\

(ii) For $c=\pm 1$, we consider $(n+1)$--dimensional space
$\R^{n+1}(x_0,x_1,\dots,x_n)$. The plane
$\tilde\pi_j=\Span\{(0,q_j),(1,Ep_j)\}$ is tangent to the cone
$$
\mathcal K_{c,\lambda^*}\colon \quad c x_0^2+\frac{\tau_1 x_1^2}{\lambda^*-U_1}+\dots+\frac{\tau_n x_n^2}{\lambda^*-U_n}=0
$$
if and only if
$$
\det\mathcal
L_j(\lambda^*)=Q_{\lambda^*}(q_j,q_j)(c+Q_{\lambda^*}(Ep_j,Ep_j))-Q_{\lambda^*}(q_j,Ep_j)^2=0.
$$
Here $\mathcal L_j(\lambda)$ is given by Theorem \ref{sedma} with
$\epsilon=\infty$. Thus, as in item (i), if $\tilde\pi_j$ is
tangent to $\mathcal K_{c,\lambda^*}$, it is tangent to $\mathcal
K_{c,\lambda^*}$ for all $j\in\mathbb Z$. Now, the statement
follows from the identities
\begin{align*}
& \mathcal Q_{c,\lambda^*} \cong \mathcal K_{c,\lambda^*}\cap\{(x_0,x_1,\dots,x_n)\in\R^{n+1}\,\vert\, x_0=1\}, \\
& l_j=Ep_j+\{ s q_j\,\vert\, s\in\R\}\cong
\tilde\pi_j\cap\{(x_0,x_1,\dots,x_n)\in\R^{n+1}\,\vert\,x_0=1\}.\qedhere
\end{align*}
\end{proof}

Obviously, item (ii) holds for the continious Neumann system
\eqref{ham1}, \eqref{ham2} as well, by replacing
$\{(q_j,p_j)\,\vert\,j\in\mathbb Z\}$ by a trajectory
$\{(q(t),p(t))\,\vert\,t\in\R\}$. For the Euclidean case it is
proved by Moser (e.g., see \cite{Mum}). The above proof is taken
from \cite{FJ}, where it is given for the Neumann systems on
Stiefel varieties.

Let us assume
$$
U_1<U_2<\dots< U_n.
$$

In the case of the Euclidean space ($k=n$), it is well known that outside coordinates hypeplanes
through $q\in\mathbb E^{n}$ it pass exactly $n$, i.e., $n-1$ quadrics from the confocal family \eqref{confocal}, for $c=\pm 1$ and $c=0$, respectively. They define ellipsoidal coordinates, i.e., together with $r=\sqrt{\langle q,q\rangle}$ so called sphero--conical  coordinates in $\mathbb E^{n}$. Suppose $0<k<n$.

\begin{thm}\label{dvanaesta}  For $c=1$ through a generic point in $\mathbb E^{k,l}$ pass $n$  quadrics from the pseudo-confocal family \eqref{confocal},
for $c=-1$ pass $n$ or $n-2$, and for $c=0$ exactly $n-1$
quadrics.
\end{thm}
\begin{proof}
For $c=1$ and $A=EU$, the confocal family \eqref{confocal} corresponds to the confocal family studied in \cite{JJ3}.
Consider the confocal family \eqref{confocal} written in the
form
\begin{equation}\label{confeq}
c+Q_{\lambda}(q,q)=c+\sum_{i=1}^k\frac{q_i^2}{\lambda-U_i}-\sum_{i=k+1}^n\frac{q_i^2}{\lambda-U_i}=\frac{c\lambda^{n}+\cdots}{(\lambda-U_1)\dots\cdot(\lambda-U_n)}=0.
\end{equation}

From $\lim_{\lambda\rightarrow
U_i\pm}{q_i^2}/(\lambda-U_i)=\pm\infty$ we see that there exist at
least $n-2$ solutions
\[
\zeta_1\in(U_1,U_2),\dots,
\zeta_{k-1}\in(U_{k-1}, U_k), \zeta_{k+1}\in(U_{k+1},
U_{k+2}),\dots, \zeta_{n-1}\in(U_{n-1}, U_n)
\]
of \eqref{confeq}
outside coordinates hyperplanes.

Next, from
\begin{equation}\label{Qas}
Q_{\lambda}(q,q)=\frac{\lambda^{n-1}\langle
q,q\rangle+\cdots}{(\lambda-U_1)\dots\cdot(\lambda-U_n)} \sim \frac{\langle
q,q\rangle}{\lambda}, \qquad \lambda\rightarrow\pm\infty,
\end{equation}
and
\[
,\quad \lim_{\lambda\rightarrow U_1-}{q_1^2}/(\lambda-U_1)=-\infty,
\quad \lim_{\lambda\rightarrow U_n+}{-q_n^2}/(\lambda-U_n)=-\infty,
\]
it follows that in the case $c=1$ there are two additional
solutions $\zeta_0\in(-\infty, U_1)$ and $\zeta_n\in(U_n,\infty)$.

Further, for $c=0$ and $\langle q,q\rangle<0$, from \eqref{Qas}, we have
a solution $\zeta_0$ within $(-\infty, U_1)$, while for
$\langle q,q\rangle>0$ we have a solution $\zeta_n\in (U_n,\infty)$.
\end{proof}
\begin{thm} \label{trinaesta}
For $c=1$ and a generic trajectory $\{(q_j,p_j)\,\vert\,j\in\mathbb Z\}$, the sequence of lines \eqref{LINES} is
tangent to $n-1$ quadrics from the pseudo--confocal family \eqref{confocal}, while
for $c=0$, the sequence of planes \eqref{PLANES} is tangent to  $n-2$ cones $\mathcal Q_{0,\lambda}$.
For $c=-1$ and a generic trajectory $\{(q_j,p_j)\,\vert\,j\in\mathbb Z\}$, the sequence of lines \eqref{LINES} is, depending on the initial position, tangent to $n-1$ or $n-3$ quadrics.
\end{thm}

\begin{proof}
According to Theorem \ref{jedanaesta},
we need to estimate the number of the real zeros of the equation $\mathcal L_j(\lambda)=0$.
To simplify the notation, in what follows we will omit the index $j$ and use $p$, $q$, $\mathcal L(\lambda)$, instead of $p_j$, $q_j$, and $\mathcal L_j(\lambda)$.

Recall the equation \eqref{DETc} and rewrite it as
\begin{equation}\label{determ}
Q_\lambda(q,q)(c+Q_\lambda(Ep,Ep))-Q_\lambda(q,Ep)^2=
\sum_{i=1}^n\frac{f_i(q,p)}{\lambda-U_i}=\frac{P_c(\lambda)}{\prod_i(\lambda-U_i)},
\end{equation}
where $f_i$ are the integrals \eqref{int-c} and $P_c(\lambda)$ is a polynomial of degree $n-1$ for $c=\pm 1$ and $n-2$ for
$c=0$.
Thus, the maximal number of quadrics $\mathcal Q_{c,\lambda}$ is $n-1$ (for $c=\pm 1$), i.e, $n-2$ (for $c=0$).
Due to relations
\[
f_1+\dots+f_n=c^2, \quad
U_1 f_1+\cdots+ U_n f_n =-\langle Ep,q\rangle^2 \quad (\text{for  } c=0),
\]
the leading terms of polynomials $P_c(\lambda)$ are given by
\begin{equation}\label{EXP}
P_{\pm 1}(\lambda)=\lambda^{n-1} + \cdots, \qquad P_0(\lambda)=-\langle Ep,q\rangle^2 \lambda^{n-2}+\cdots.
\end{equation}

Firstly,  let us assume $c=\langle
q,q\rangle=-1$, $q_1\dots q_n\ne 0$. As in the proof of Theorem \ref{dvanaesta}, there are
$n-1$ solutions
\[
\zeta_0\in(-\infty, U_1),\dots,\zeta_{k-1}\in(U_{k-1}, U_k), \zeta_{k+1}\in(U_{k+1}, U_{k+2}),\dots, \zeta_{n-1}\in(U_{n-1}, U_n)
\]
of the equation $Q_{\lambda}(q,q)=0$.

The left hand side of \eqref{determ} is negative at the ends of all $n-3$ intervals
\begin{equation}\label{interval}
(\zeta_0,\zeta_1), (\zeta_1,\zeta_2),\dots, (\zeta_{k-2}, \zeta_{k-1}), (\zeta_{k+1}, \zeta_{k+2}),\dots, (\zeta_{n-2}, \zeta_{n-1}),
\end{equation}
which contain $U_1, U_2,\dots, U_{k-1}, U_{k+2},\dots, U_{n-1}$, respectively. Owing to
\begin{equation}\label{limF}
\lim_{\lambda\rightarrow U_i\pm}\frac{f_i}{\lambda-U_i}=(\pm \mathrm{sgn} f_i)\cdot\infty,
\end{equation}
we see that each interval in \eqref{interval} contains a solution of $\det \mathcal
L(\lambda)=0$.

In the case $c=\langle q,q\rangle=1$, with $U_1$ and $(\zeta_0,\zeta_1)$ replaced by $U_n$ and $(\zeta_{n-1}, \zeta_{n})$, we get the existence of $n-3$ solutions of $\det \mathcal
L(\lambda)=0$.
On the other side, from \eqref{EXP} we get the asymptotic expansion
\[
\sum_{i=1}^n\frac{f_i(q,p)}{\lambda-U_i} \sim  \frac{1}{\lambda}, \qquad \lambda\rightarrow\pm\infty,
\]
leading to a solution  within $(\zeta_n,\infty)$ as well.
Since the polynomial $P_{1}(\lambda)$ has
real coefficients, degree $n-1$, and $n-2$ real zeros (none of the given zeros is of multiplicity greater then 1), it has an additional real zero.

For the case $c=\langle q,q\rangle=0$ we proceed analogously.
As in the proof of Theorem \ref{dvanaesta}, there are
$n-2$ solutions
\[
\zeta_1\in(U_1, U_2),\dots,\zeta_{k-1}\in(U_{k-1}, U_k), \zeta_{k+1}\in(U_{k+1}, U_{k+2}),\dots, \zeta_{n-1}\in(U_{n-1}, U_n)
\]
of the equation $Q_{\lambda}(q,q)=0$.
The left hand side of \eqref{determ} is negative at the ends of all $n-4$ intervals
\begin{equation}\label{interval**}
(\zeta_1,\zeta_2),\dots, (\zeta_{k-2}, \zeta_{k-1}), (\zeta_{k+1}, \zeta_{k+2}),\dots, (\zeta_{n-2}, \zeta_{n-1}),
\end{equation}
that contain $U_2, U_3,\dots, U_{k-1}, U_{k+2},\dots, U_{n-1}$. From \eqref{limF},
we obtain that each interval in \eqref{interval**} contains a solution of $\det \mathcal
L(\lambda)=0$. Moreover, due to \eqref{EXP}, we have the asymptotic expansion
\[
\sum_{i=1}^n\frac{f_i(q,p)}{\lambda-U_i}\sim
-\frac{\langle Ep,q\rangle^2}{\lambda^2}, \qquad \lambda\rightarrow\pm\infty
\]
implying that there exist $\zeta_0< U_1$ and $\zeta_n>U_n$, such that the left hand side of \eqref{determ}
is less then zero.

Therefore, the equation $\det \mathcal
L(\lambda)=0$ has $n-2$ real solutions.
\end{proof}

\begin{rem}{\rm
The signatures $(1,n-1)$ and $(n-1,1)$ should be treated separately,
however for $c=1$, $c=0$, and $c=-1$ and the signature $(1,n-1)$ the conclusions are the same.
Suppose $c=-1$ and $k=n-1$.
Now the left hand side of \eqref{determ} is negative at the ends of intervals
\begin{equation*}
(\zeta_0,\zeta_1), (\zeta_1,\zeta_2),\dots, (\zeta_{n-3}, \zeta_{n-2})
\end{equation*}
that contain $U_1, U_2,\dots, U_{n-3}, U_{n-2}$,
so we get $n-2$ real solutions of
$\det \mathcal L(\lambda)=0$.
Again, since  $P_{-1}(\lambda)$ has
$n-2$ real zeros, it has the additional real zero: the sequence of lines \eqref{LINES} is
tangent to $n-1$ quadrics $\mathcal Q_{-1,\lambda}$ for a generic initial conditions.
}\end{rem}

\begin{rem}{\rm
If we assume $c=-1$ and that the value of the integal $f_k$ is less then zero or the value of $f_{k+1}$ is greater then zero, then  the sequence of lines \eqref{LINES} is
tangent to $n-1$ quadrics $\mathcal Q_{-1,\lambda}$.
Indeed, then, from \eqref{limF}, there exists
$\zeta_{k}\in (U_k,U_{k+1})$ with $\det \mathcal
L(\zeta_{k})<0$. Since
\[
U_k\in (\zeta_{k-1}, \zeta_k) \quad  \text{and} \quad  U_{k+1}\in (\zeta_k, \zeta_{k+1})  \quad  (\text{for  } k<n-1),
\]
there exist two additional real solutions of $\det \mathcal L(\lambda)=0$.
}\end{rem}

\begin{rem}
{\rm
Note that one can consider a symmetric Heisenberg model, i.e., Neumann system on $S^{n-1}_c$ as well, when some of $U_i$ are mutually equal:
\[
\underbrace{U_1=\cdots=U_{\rho_1}}_{\rho_1} <
\underbrace{U_{\rho_1+1}=\cdots U_{\rho_1+\rho_2}}_{\rho_2}<\cdots<
\underbrace{U_{n-\rho_r+1}=\cdots=U_{n}}_{\rho_r},
\]
$\rho_1+\cdots+\rho_r=n$. Then the set of all symmetries
$
\mathbf R\in O(k,l)\colon \Ad_\mathbf R(U)=U
$
is either
$
O(\rho_1)\times\cdots\times O(\rho_r),
$
or
\[
O(\rho_1)\times\cdots O(\rho_{p-1})\times O(k_p,l_p)\times O(\rho_{p+1})\times \cdots\times O(\rho_r), \quad k_p+l_p=\rho_p.
\]

Similarly like in the case of virtual billiard dynamics \cite{JJ3}, the systems are integrable in a noncommutative sense and the phase spaces $\mathfrak M_c$, $c=\pm 1, 0$ are foliated on invariant $(N-1)$--dimensional isotropic varieties, where
\[
N=r+\sharp\{s\,\vert\, \rho_s>1, \, s=1,\dots,r\}.
\]

Further, some additional careful analysis is
needed in order to estimate the number of real caustics and their maximal number is $N-1$ for $c=\pm 1$ and $N-2$ for $c=0$.
}
\end{rem}

\section{Appendix}

\begin{proof}[Proof of Lemma \ref{karakterizacija}]
Let
$
S(x,X)=\langle B x,X\rangle,
$
where $B$ is a nonsingular matrix. The equations \eqref{GF2} become
\begin{align}
y=B^T X+\lambda \A x, \qquad Y=-Bx-\Lambda\A X,\label{1nova}
\end{align}
where $\langle \A x,x\rangle=1$, $\langle \A X,X\rangle=1$.

From the constraints $\langle y,y\rangle=1$, $\langle Y,Y\rangle=1$, we get that  $\lambda$ and $\Lambda$ are
solutions of the equations
\begin{align}
\label{D1} 1&=\langle B^T X,B^T X\rangle +2\lambda \langle\A x,B^T X\rangle+\lambda^2 \langle A^{-2}x,x\rangle,\\
\label{D2} 1&=\langle B x, B x\rangle+2\Lambda \langle \A X,B x\rangle+\Lambda^2 \langle A^{-2}X,X\rangle.
\end{align}

One can easily see that if
$$
\max_{\xi\in\mathbb Q^{n-1}} \vert B^T \xi\vert=\max_{\xi\in \mathbb S^{n-1}}\vert B^T A^{1/2} \xi\vert >1,
\qquad
\max_{\xi\in\mathbb Q^{n-1}} \vert B \xi\vert=\max_{\xi\in \mathbb S^{n-1}}\vert B A^{1/2} \xi\vert >1,
$$
then
there exists  $(x,X)\in \mathbb Q^{n-1}\times\mathbb Q^{n-1}$ such that the discriminant of \eqref{D1}, respectively \eqref{D2}, is less then zero. On the other hand,
if $\max_{\xi\in \mathbb S^{n-1}}\vert B^T A^{1/2} \xi\vert  \le 1$ and  $\max_{\xi\in \mathbb S^{n-1}}\vert B A^{1/2} \xi\vert  \le 1$, the discriminants are greater then zero and
we have real multipliers as functions on $\mathbb Q^{n-1}\times\mathbb Q^{n-1}(x,X)$. Further, if the above
relations define the mapping $\psi\colon M_{1,1}\to M_{1,1}$,
we have
\begin{align}
X&=(B^T)^{-1} y-\nu (AB^T)^{-1} x,              \label{1nova+}\\
Y&=-Bx-\mu \A X=-Bx-\mu\A((B^T)^{-1} y-\nu (AB^T)^{-1}x), \label{2nova+}
\end{align}
for some multipliers $\nu$, $\mu$, now functions on $M_{1,1}(x,y)$.
From \eqref{1nova+} and the constraint $\langle\A X,X\rangle=1$,
we get
\begin{align*}
\nu^2\langle \A (AB^T)^{-1} x,(AB^T)^{-1} x\rangle
-2\nu\langle
(B^TA)^{-1}y,(AB^T)^{-1}x \rangle+
 \vert (B^TA^{1/2})^{-1}y\vert^2=1.
\end{align*}

Again, if
$$
\max_{\xi\in \mathbb S^{n-1}}\vert (B^TA^{1/2})^{-1}\xi\vert=\max_{\xi\in \mathbb S^{n-1}} 1/\vert B^TA^{1/2}\xi\vert >1,
$$
there exists  $(x,y)\in M$ such that the discriminant of the above quadratic equation is less then zero.
Thus, in that case, \eqref{1nova} defines a dynamics for complexified objects only.
Therefore, we obtain the necessary condition
$\vert B^T A^{1/2}\vert=1$.
A similar analysis for $(x,y)$ to be expressed as functions of $(X,Y)$, leads to the condition $\vert B A^{1/2}\vert=1$.
\end{proof}


\begin{proof}[Proof of Theorem \ref{treca}] We have

\begin{align*}
\mathcal A_k\mathcal L_k &=A^2+\lambda K_1+\lambda^2 K_2+\lambda^3 c\,
K_3+\lambda^4 c^2 K_4, \\
\mathcal L_{k+1}\mathcal A_k &=A^2+\lambda
S_1+\lambda^2 S_2+\lambda^3 c\, S_3+\lambda^4 c^2 S_4,
\end{align*}
where
\begin{align*}
K_1 =&  Fy_k \otimes A F x_k-Fx_k\otimes A F y_{k-1}+ A F
y_{k-1}\otimes Fx_k-A F x_k\otimes Fy_{k-1},\\
K_2 =&  - c\, Fy_k \otimes AF y_{k-1}- c\,AFy_{k-1}\otimes
Fy_{k-1}+\langle y_{k-1},x_k\rangle Fx_k\otimes Fy_{k-1}\\
& -\langle y_{k-1},y_{k-1}\rangle Fx_k\otimes Fx_k+\langle
x_k,y_{k-1}\rangle Fy_k\otimes Fx_k-\langle x_k,x_k\rangle
Fy_k\otimes Fy_{k-1},\\
K_3 =& \langle y_{k-1},y_{k-1}\rangle Fx_k\otimes
Fy_{k-1}-\langle
y_{k-1},y_{k-1}\rangle Fy_k\otimes Fx_k,\\
K_4 =& \langle y_{k-1},y_{k-1}\rangle Fy_k\otimes Fy_{k-1},
\end{align*}
and
\begin{align*}
S_1 =&  Fy_k \otimes A F x_{k+1}-Fx_{k+1}\otimes A F y_{k}+ A F
y_{k}\otimes Fx_k-A F x_k\otimes Fy_{k-1},\\
S_2 =&  - c\, Fy_k \otimes A F y_{k}-c\, A Fy_{k}\otimes
Fy_{k-1}+\langle y_{k},x_k\rangle Fx_{k+1}\otimes Fy_{k-1}\\
& -\langle x_{k+1},x_{k}\rangle Fy_k\otimes Fy_{k-1}-\langle
y_k,y_{k}\rangle Fx_{k+1}\otimes Fx_k+\langle x_{k+1},y_k\rangle
Fy_k\otimes Fx_{k},\\
S_3 =& \langle y_{k},y_{k}\rangle Fx_{k+1}\otimes
Fy_{k-1}-\langle
x_{k+1},y_{k}\rangle Fy_k\otimes Fy_{k-1}\\
&+\langle y_{k},x_{k} \rangle Fy_{k}\otimes Fy_{k-1}-\langle
y_{k},y_{k}\rangle Fy_k\otimes Fx_{k}, \\
S_4 =& \langle y_{k},y_{k}\rangle Fy_k\otimes Fy_{k-1}.
\end{align*}

It is evident that $K_4=S_4$. From \eqref{1bilijar},
\eqref{2bilijar} we obtain
\begin{align*}
K_1-S_1 =&  Fy_k \otimes A F x_k-Fx_k\otimes A F y_{k-1}+ A F
y_{k-1}\otimes Fx_k \\
& -Fy_k \otimes A F(x_k+\mu_ky_k)+F(x_k+\mu_ky_k)\otimes A F
y_{k}- A F
y_{k}\otimes Fx_k\\
 = &(A F y_{k-1}-A F y_k) \wedge Fx_k= -\nu_{k-1} A F A^{-1} x_k\wedge Fx_k=0,
\end{align*}
\begin{align*}
K_2-S_2 =& c\,\big( Fy_k \otimes (A F y_{k}- A F y_{k-1}) +(A Fy_{k}
-A Fy_{k-1})\otimes Fy_{k-1}\big) \\
&+\langle y_{k-1},x_k\rangle Fx_k\otimes Fy_{k-1}-\langle
y_{k-1},y_{k-1}\rangle Fx_k\otimes Fx_k\\
&+\langle x_k,y_{k-1}\rangle Fy_k\otimes Fx_k-\langle
x_k,x_k\rangle
Fy_k\otimes Fy_{k-1}\\
& -\langle y_{k},x_k\rangle F(x_k+\mu_ky_k)\otimes
Fy_{k-1}+\langle x_k+\mu_ky_k,x_{k}\rangle Fy_k\otimes
Fy_{k-1}\\
&+\langle y_k,y_{k}\rangle F(x_k+\mu_ky_k)\otimes Fx_k-\langle
x_k+\mu_ky_k,y_k\rangle
Fy_k\otimes Fx_{k},
\end{align*}
that is,
\begin{align*}
K_2-S_2=&c\,\big(\nu_{k-1}Fy_k \otimes Fx_k  +\nu_{k-1} Fx_k\otimes Fy_{k-1}\big)\\
&+\langle y_{k-1}-y_k,x_k\rangle Fx_k\otimes Fy_{k-1}+\langle
x_k,y_{k-1}-y_k\rangle Fy_k\otimes Fx_k\\
=&c\,\big(\nu_{k-1}Fy_k \otimes Fx_k  +\nu_{k-1} Fx_k\otimes Fy_{k-1}\big)\\
&+\langle -\nu_{k-1}A^{-1}x_k,x_k\rangle Fx_k\otimes
Fy_{k-1}+\langle x_k,-\nu_{k-1}A^{-1}x_k\rangle Fy_k\otimes
Fx_k=0,
\end{align*}
and
\begin{align*}
K_3-S_3 =& \langle y_{k-1},y_{k-1}\rangle Fx_k\otimes
Fy_{k-1}-\langle y_{k-1},y_{k-1}\rangle Fy_k\otimes Fx_k\\
&-\langle y_{k},y_{k}\rangle F(x_k+\mu_ky_k)\otimes
Fy_{k-1}+\langle
x_k+\mu_ky_k,y_{k}\rangle Fy_k\otimes Fy_{k-1}\\
&-\langle y_{k},x_{k} \rangle Fy_{k}\otimes Fy_{k-1}+\langle
y_{k},y_{k}\rangle Fy_k\otimes Fx_{k}=0.\qedhere
\end{align*}
\end{proof}

\begin{proof}[Proof of Lemma \ref{komutativna}]
Let $\mathbb L_1^{-1}(q,p)=\{(q,Q_1),(q,Q_2)\}$,$(Q_i,\bar Q_i)=\Phi(q,Q_i)$, $(Q_i,P_i)=\mathbb L_2(q,Q_i)$, $(Q_i,\bar P_i)=\mathbb L_1(Q_i,\bar Q_i)$,
i.e.,
\begin{eqnarray*}
&& p=EJQ_i- \frac{1}c \langle JQ_i,q\rangle Eq, \\
&& P_i=-EJq+\frac{1}c \langle JQ_i,q\rangle EQ_i,\\
&& \bar Q_i=-q+2\frac{\langle J^{-1}q,Q_i\rangle}{\langle J^{-2} Q_i,Q_i\rangle} J^{-1}Q_i, \\
&& \bar P_i=EJ\bar Q_i- \frac{1}c \langle J\bar Q_i,Q_i\rangle E Q_i, \qquad i=1,2.
\end{eqnarray*}

Now, Lemma follows from the identity
\begin{align*}
\bar P_i &=-EJq+2\frac{\langle J^{-1}q,Q_i\rangle}{\langle J^{-2} Q_i,Q_i\rangle} E Q_i- \frac{1}c \langle -Jq+2\frac{\langle J^{-1}q,Q_i\rangle}{\langle J^{-2} Q_i,Q_i\rangle}Q_i,Q_i\rangle E Q_i\\
&=-EJq+\frac{1}c \langle JQ_i,q\rangle EQ_i=P_i.\qedhere
\end{align*}
\end{proof}

\begin{proof}[Proof of Theorem \ref{deseta}]

(i) From \eqref{n01} and \eqref{n02} we have
\begin{align*}
\langle E q_{k+1},p_{k+1}\rangle =&\langle J^{-1}p_k+\mu_k E J^{-1}q_k,  -EJq_{k}+\tilde\mu_k Eq_{k+1}\rangle\\
=&-\langle E q_k,p_k\rangle -\mu_k\langle q_k,q_k\rangle+\mu_k\tilde{\mu}_k\langle J^{-1}q_k,q_{k+1}\rangle
+\tilde{\mu}_k\langle J^{-1}p_k,Eq_{k+1}\rangle\\
=&-\langle E q_k,p_k\rangle +\tilde{\mu}_k \langle EJ^{-1}p_k+\mu_k J^{-1}q_k,q_{k+1}\rangle \\
=&-\langle E q_k,p_k\rangle +\tilde{\mu}_k\langle q_{k+1},q_{k+1}\rangle=-\langle E q_k,p_k\rangle.
\end{align*}

(ii) Consider the graph $\Gamma_\kappa$ of the correspondence $\Psi\vert_{\Sigma_k}$:
$$
\Gamma_\kappa \subset \Sigma_\kappa(q,p)\times\Sigma_\kappa(Q,P) \subset \R^{2n}(q,p)\times\R^{2n}(Q,P)
$$

Note that the generating function $S=\langle q,J Q\rangle$ of the mappings $\mathbb L_1^0$, $\mathbb L_2^0$ satisfies $S=\pm \kappa\vert_{\Gamma_\kappa}$, i.e, $dS\vert_{\Gamma_\kappa}=0$. Therefore,
\begin{align*}
PdQ-pdq &= -dS+\tilde\mu EQdQ+\mu Eq dq=\frac{\tilde\mu}{2}d\langle
Q,Q\rangle+\frac{\mu}{2}d\langle q,q\rangle=0\vert_{\Gamma_\kappa}.
\end{align*}

Thus,  the Heisenberg  system restricted to $\Sigma_\kappa$ preserves
the 1-form $\theta=pdq\vert_{\Sigma_\kappa}$:
$$
(\Psi\vert_{\Sigma_\kappa})^*\theta=\theta.
$$

The Hamiltonian flow of $\varphi_2=\langle E p,q\rangle$ with
respect to the Dirac--Poisson bracket $\{\cdot,\cdot\}_D^0$ equals
\[
X=\sum_i q_i {\partial}/{\partial q_i}-p_i
{\partial}/{\partial p_i}.
\]
The submanifold $\Sigma_\kappa$ is a contact manifold with respect to
$\theta$, if and only if $\theta(X)\ne 0$ on $\Sigma_\kappa$ (e.g., see
\cite{LM}). We have $\theta(X)=\varphi_2=\pm \kappa$. Therefore,
$\Sigma_\kappa$ is a contact manifold with respect to $\theta$ for
$\kappa\ne 0$ with the Reeb vector field $Z=\pm\frac{1}{\kappa} X$.

Next, since $\{\varphi_2,f_i\}^0_D=0$, $i=1,\dots,n$, where
integrals $f_i$ are given by \eqref{INTEGRALI}, from $\{f_i,f_j\}^0_D=0$, using the theorem on {\it isoenergetic
integrability} (see \cite{JJ2}), we get that the restrictions of
$f_i\vert_{\Sigma_\kappa}$ commute with respect to the Jacobi bracket
on $(\Sigma_\kappa,\theta)$.
We have
two relations
\[
\sum_i f_i\vert_{\Sigma_\kappa}=0,\qquad
\sum_i J_i^2 f_i\vert_{\Sigma_\kappa}+\kappa^2=0,
\]
and there are
$n-2$ independent integrals on $\Sigma_\kappa$. Thus, the mapping $\Psi$
is a completely integrable contact 1:2 correspondence (see \cite{KT2,
JJ3}).
\end{proof}

\subsection*{Acknowledgments}
The research of B. J. was supported by the Serbian Ministry of
Science Project 174020, Geometry and Topology of Manifolds,
Classical Mechanics and Integrable Dynamical Systems.

\

\end{document}